\documentclass[aps,pra,groupedaddress,showpacs,twocolumn,superscriptaddress]{revtex4}

\usepackage{amsmath}
\usepackage{amsthm}

\newcommand{\ket}[1]{|#1\rangle}             
\newcommand{\bra}[1]{\langle #1|}            
\newcommand{\dyad}[2]{\ket{#1}\bra{#2}}      

\newcommand{\HC}{\mathcal{H}}
\newcommand{\SC}{\mathcal{S}}
\newcommand{\TC}{\mathcal{T}}

\newtheorem{theorem}{Theorem}
\newtheorem{corollary}{Corollary}

\begin{document}

\title{Entanglement transformations using separable operations}

\author{Vlad Gheorghiu}
\email[Electronic address:]{vgheorgh@andrew.cmu.edu}

\author{Robert B. Griffiths}
\affiliation{Department of Physics, Carnegie Mellon University, Pittsburgh,
Pennsylvania 15213, USA}

\date{Version of 18 September 2007}

\begin{abstract}
We study conditions for the deterministic transformation
$\ket{\psi}\longrightarrow\ket{\phi}$ of a bipartite entangled state by a
separable operation. If the separable operation is a local operation with classical communication (LOCC), Nielsen's
majorization theorem provides necessary and sufficient conditions. For the
general case we derive a necessary condition in terms of products of Schmidt
coefficients, which is equivalent to the Nielsen condition when either of
the two factor spaces is of dimension 2, but is otherwise weaker.  One
implication is that no separable operation can reverse a deterministic map
produced by another separable operation, if one excludes the case where
the Schmidt coefficients of $\ket{\psi}$ and are the same as those of
$\ket{\phi}$.  The question of sufficient conditions in the general
separable case remains open. When the Schmidt coefficients of $\ket{\psi}$
are the same as those of $\ket{\phi}$, we show that the Kraus operators of the
separable transformation restricted to the supports of $\ket{\psi}$ on the factor
spaces are proportional to unitaries.  When that proportionality holds and
the factor spaces have equal dimension, we find conditions for the
deterministic transformation of a collection of several full Schmidt rank
pure states $\ket{\psi_j}$ to pure states $\ket{\phi_j}$.
\end{abstract}

\maketitle

\section{Introduction}
\label{sct1}

A separable operation $\Lambda$ on a bipartite quantum system is a
transformation of the form
\begin{equation}
\label{eqn1} \rho' = \Lambda(\rho)=\sum_m
\bigl(A_m^{}\otimes B_m^{}\bigr)\rho (A_m^\dagger \otimes B_m^\dagger\bigr),
\end{equation}
where $\rho$ is an initial density operator on the Hilbert space
$\HC_A\otimes\HC_B$.  The Kraus operators $A_m \otimes B_m$ are
arbitrary product operators satisfying the closure condition
\begin{equation}
\label{eqn2}
\sum_m A_m^\dagger A_m^{}\otimes B_m^\dagger B_m^{}=I\otimes I.
\end{equation}
The extension of \eqref{eqn1} and \eqref{eqn2} to multipartite systems is
obvious, but here we will only consider the bipartite case.  To avoid
technical issues the sums in \eqref{eqn1} and \eqref{eqn2} and the dimensions
of $\HC_A$ and $\HC_B$ are assumed to be finite.

Various kinds of separable operations play important roles in quantum
information theory. When $m$ takes on only one value the operators $A_1$ and
$B_1$ are (or can be chosen to be) unitary operators, and the operation is a
\emph{local unitary} transformation. When every $A_m$ and every $B_m$ is
proportional to a unitary operator, we call the operation a \emph{separable
  random unitary channel}. Both of these are members of the well-studied class
of \emph{local operations with classical communication} (LOCC), which can be
thought of as an operation carried out by Alice on $\HC_A$ with the outcome
communicated to Bob.  He then uses this information to choose an operation
that is carried out on $\HC_B$, with outcome communicated to Alice, who uses it
to determine the next operation on $\HC_A$, and so forth.  For a precise
definition and a discussion, see [\cite{HorodeckiQuantumEntanglement}, Sec.~XI]. 
While any LOCC is a separable operation,
i.e., can be written in the form \eqref{eqn1}, the reverse is not true: there
are separable operations which fall outside the LOCC class
\cite{PhysRevA.59.1070}.

Studying properties of general separable operations seems worthwhile because
any results obtained this way then apply to the LOCC subcategory, which is
harder to characterize from a mathematical point of view.  However, relatively
little is known about separable operations, whereas LOCC has been the subject
of intensive studies, with many important results.  For example, an LOCC
applied to a pure entangled state $\ket{\psi}$ (i.e., $\rho=\dyad{\psi}{\psi}$
in \eqref{eqn1}) results in an ensemble of pure states (labeled by $m$) whose
average entanglement cannot exceed that of $\ket{\psi}$, [\cite{HorodeckiQuantumEntanglement}, Sec.~XV D].  One suspects that the same is true of a
general separable operation $\Lambda$, but this has not been proved.  All that
seems to be known is that $\Lambda$ cannot ``generate'' entanglement when
applied to a product pure state or a separable mixed state: the outcome (as is
easily checked) will be a separable state.

If an LOCC is applied to a pure (entangled) state $\ket{\psi}$, Lo and Popescu
\cite{PhysRevA.63.022301} have shown that the same result, typically an
ensemble, can be achieved using a different LOCC (depending both on the
original operation and on $\ket{\psi}$) in which Alice carries out an
appropriate operation on $\HC_A$ and Bob a unitary, depending on that outcome,
on $\HC_B$.  This in turn is the basis of a condition due to Nielsen
\cite{Nielsen:Majorization} which states that there is an LOCC operation
deterministically (probability 1) mapping a given bipartite state
$\ket{\psi}$ to another pure state $\ket{\phi}$ if and only if $\ket{\phi}$
majorizes $\ket{\psi}$ \cite{ntk1}.

In this paper we derive a necessary condition for a separable operation to
deterministically map $\ket{\psi}$ to $\ket{\phi}$ in terms of their Schmidt
coefficients, the inequality \eqref{eqn5}.  While it is weaker than Nielsen's
condition (unless either $\HC_A$ or $\HC_B$ is two dimensional, in which case
it is equivalent), it is not trivial. In the particular case that the Schmidt
coefficients are the same, i.e., $\ket{\psi}$ and $\ket{\phi}$ are equivalent
under local unitaries, we show that all the $A_m$ and $B_m$ operators in
\eqref{eqn1} are proportional to unitaries, so that in this case the separable
operation is also a random unitary channel.  For this situation we
also study the conditions under which a whole \emph{collection}
$\{\ket{\psi_j}\}$ of pure states are deterministically mapped to pure states,
a problem which seems not to have been previously studied either for LOCC or
for more general separable operations.

The remainder of the paper is organized as follows.  Section~\ref{sct2} has
the proof, based on a inequality by Minkowski, p.~482 of
\cite{HornJohnson:MatrixAnalysis}, of the relationship between the Schmidt
coefficients of $\ket{\psi}$ and $\ket{\phi}$ when a separable operation
deterministically maps $\ket{\psi}$ to $\ket{\phi}$, and some consequences of
this result.  In Section~\ref{sct3} we derive and discuss the conditions under
which a separable random unitary channel will map a collection of pure
states to pure states.  A summary and some discussion of open questions will
be found in Section~\ref{sct4}.

\section{Local transformations of bipartite entangled states}
\label{sct2}

We use the term \emph{Schmidt coefficients} for the \emph{nonnegative}
coefficients $\{\lambda_j\}$ in the Schmidt expansion
\begin{equation}
\label{eqn3}
\ket{\psi}=\sum_{j=0}^{d-1}\lambda_j\ket{\bar a_j}\otimes\ket{\bar b_j},
\end{equation}
of a state $\ket{\psi}\in \HC_A\otimes\HC_B$, using appropriately chosen
orthonormal bases $\{\ket{\bar a_j}\}$ and $\{\ket{\bar b_j}\}$, with the
order chosen so that
\begin{equation}
\label{eqn4}
 \lambda_0 \geq \lambda_1 \geq \cdots \geq \lambda_{d-1} \geq 0.
\end{equation}
The number $r$ of positive (nonzero) Schmidt coefficients is called the
\emph{Schmidt rank}.  We call the subspace of $\HC_A$ spanned by $\ket{\bar
  a_0},\ket{\bar a_1}\ldots \ket{\bar a_{r-1}}$, i.e., the basis kets for which
the Schmidt coefficients are positive, the $\HC_A$ \emph{support} of
$\ket{\psi}$, and that spanned by $\ket{\bar b_0},\ket{\bar b_1}\ldots \ket{\bar b_{r-1}}$ its
$\HC_B$ \emph{support}.

 Our main result is the following:
\begin{theorem}
\label{thm1}
Let $\ket{\psi}$ and $\ket{\phi}$ be two bipartite entangled states on
$\mathcal{H}_A\otimes\mathcal{H}_B$ with positive Schmidt coefficients
$\{\lambda_j\}$ and $\{\mu_j\}$, respectively, in decreasing order, and let
$r$ be the Schmidt rank of $\ket{\psi}$.  If
$\ket{\psi}$ can be transformed to $\ket{\phi}$ by a deterministic separable
operation, then

i) The Schmidt rank of $\ket{\phi}$ is less than or equal to $r$.

ii)
\begin{equation}
\label{eqn5}
\prod_{j=0}^{r-1}\lambda_j\geq\prod_{j=0}^{r-1}\mu_j.
\end{equation}

iii) If \eqref{eqn5} is an equality with both sides positive,
the Schmidt coefficients of $\ket{\psi}$ and $\ket{\phi}$ are identical,
$\lambda_j = \mu_j$, and the operators $A_m$ and $B_m$ restricted to the
$\HC_A$ and $\HC_B$ supports of $\ket{\psi}$, respectively, are proportional
to unitary operators.

iv) The reverse deterministic transformation of $\ket{\phi}$ to
$\ket{\psi}$ by a separable operation is only possible when the Schmidt
coefficients are identical, $\lambda_j = \mu_j$.
\end{theorem}

\begin{proof}

For the proof it is convenient to use map-state duality (see
\cite{ZcBn04,Griffiths:AtemporalDiagrams} and [\cite{GeometryQuantumStates}, Chap.~11])
defined in the following way.  Let $\{\ket{b_j}\}$ be an orthonormal basis
of $\HC_B$ that will remain fixed throughout the following discussion.
Any ket $\ket{\chi}\in\HC_A\otimes\HC_B$ can be expanded in this basis in the
form
\begin{equation}
\label{eqn6}
  \ket{\chi} = \sum_j \ket{\alpha_j}\otimes\ket{b_j},
\end{equation}
where the $\{\ket{\alpha_j}\}$ are the (unnormalized) expansion coefficients.
We define the corresponding dual map $\chi:\HC_B\rightarrow \HC_A$ to be
\begin{equation}
\label{eqn7}
 \chi = \sum_j\ket{\alpha_j}\bra{b_j}.
\end{equation}
Obviously any map from $\HC_B$ to $\HC_A$ can be written in the form
\eqref{eqn7}, and can thus be transformed into a ket on $\HC_A\otimes\HC_B$ by
the inverse process: replacing $\bra{b_j}$ with $\ket{b_j}$.  The
transformation depends on the choice of basis $\{\ket{b_j}\}$, but this will
not matter, because our results will in the end be independent of this choice.
Note in particular that the \emph{rank} of the operator $\chi$ is exactly the
same as the \emph{Schmidt rank} of $\ket{\chi}$.

For a separable operation that deterministically maps $\ket{\psi}$ to
$\ket{\phi}$ (or, to be more specific, $\dyad{\psi}{\psi}$ to
$\dyad{\phi}{\phi}$) it must be the case that
\begin{equation}
\label{eqn8}
\bigl(A_m\otimes B_m\bigr)\ket{\psi}=\sqrt{p_m}\ket{\phi},
\end{equation}
for every $m$, as otherwise the result of the separable operation acting on
$\ket{\psi}$ would be a mixed state.  (One could also include a complex phase
factor depending on $m$, but this can be removed by incorporating it in
$A_m$---an operation is not changed if the Kraus operators are multiplied by
phases.)  By using map-state duality we may rewrite \eqref{eqn8} in the form
\begin{equation}
\label{eqn9} A_m\psi \bar B_m = \sqrt{p_m}\phi,
\end{equation}
where by $\bar B_m$ we mean the \emph{transpose} of this operator in the basis
$\{\ket{b_j}\}$---or, to be more precise, the operator whose matrix in this
basis is the transpose of the matrix of $B_m$. From \eqref{eqn9} one sees at
once that since the rank of a product of operators cannot be larger than the
rank of any of the factors, the rank of $\phi$ cannot be greater than that of
$\psi$.  When translated back into Schmidt ranks this proves (i).

For the next part of the proof let us first assume that $\HC_A$ and
$\HC_B$ have the same dimension $d$, and that the Schmidt ranks of both
$\ket{\psi}$ and $\ket{\phi}$ are equal to $d$; we leave until later the
modifications necessary when these conditions are not satisfied.  In light of
the previous discussion of \eqref{eqn9}, we see that $\bar B_m$ has rank
$d$, so is invertible. Therefore one can solve \eqref{eqn9} for $A_m$, and if
the solution is inserted in \eqref{eqn2} the result is

\begin{widetext}
\begin{equation}
\label{eqn10}
I\otimes I= \sum_m p_m\big[\psi^{-1\dagger}
\bar B_m^{-1\dagger}(\phi^\dagger\phi)\bar B_m^{-1}\psi^{-1}\big]
\otimes \big[B_m^\dagger B_m\big]
\end{equation}
\end{widetext}

The Minkowski inequality (\cite{HornJohnson:MatrixAnalysis}, p. 482) for a sum of positive semidefinite operators on a $D$-dimensional space
is
\begin{equation}
\label{eqn11}
{\Bigg[\det\Big(\sum_m
Q_m\Big)\Bigg]}^{1/D}\geq\sum_m{\Big(\det{Q_m}\Big)}^{1/D},
\end{equation}
with equality if and only if all $Q_m$'s are proportional, i.e.
$Q_i=f_{ij}Q_j$, where the $f_{ij}$ are positive constants.  Since
$A_m^\dagger A_m\otimes B_m^\dagger B_m$ is a positive operator on a
$D=d^2$ dimensional space, \eqref{eqn10} and \eqref{eqn11} yield
\begin{widetext}
\begin{eqnarray}
\label{eqn12}
  1&\geq&{\Bigg[\det\Big(\sum_m p_m\big[\psi^{-1\dagger}
\bar B_m^{-1\dagger}(\phi^\dagger\phi){\bar B_m}^{-1}\psi^{-1}\big]
\otimes \big[B_m^\dagger B_m\big]\Big)\Bigg]}^{1/d^2}\nonumber \\
  &\geq&\sum_m {\Bigg[\det\Big(p_m\big[{\psi^{-1}}^\dagger
\bar B_m^{-1\dagger}(\phi^\dagger\phi){\bar B_m}^{-1}\psi^{-1}\big]\otimes
\big[B_m^\dagger B_m\big]\Big)\Bigg]}^{1/d^2}\nonumber\\
 &=& \sum_m p_m
\frac{\det(\phi^\dagger\phi)^{1/d}}{\det(\psi^\dagger\psi)^{1/d}} =
\frac{\det(\phi^\dagger\phi)^{1/d}}{\det(\psi^\dagger\psi)^{1/d}},
\end{eqnarray}
\end{widetext}
which is equivalent to
\begin{equation}
\label{eqn13}
\det(\psi^\dagger\psi)\geq\det(\phi^\dagger\phi).
\end{equation}
The relation $\det(A\otimes B)$=$(\det A)^b (\det B)^a$ , where
$a,b$ are the dimensions of $A$ and $B$, was used in deriving
\eqref{eqn12}.
Since \eqref{eqn13} is the square of \eqref{eqn5}, this proves part (ii).

If \eqref{eqn5} is an equality with both sides positive,
$\det(\phi^\dagger\phi)/\det(\psi^\dagger\psi)=1$ and the inequality
\eqref{eqn12} becomes an equality, which implies that all positive
operators in \eqref{eqn11} are proportional, i.e.
\begin{equation}
\label{eqn14}
A_m^\dagger A_m^{}\otimes B_m^\dagger B_m^{} =
f_{mn}^{}A_n^\dagger A_n^{}\otimes B_n^\dagger B_n^{},
\end{equation}
where the $f_{mn}$ are positive constants. Setting $n=1$ in
\eqref{eqn14} and inserting it in \eqref{eqn2} one gets
\begin{equation}
\label{eqn15}
(\sum_m f_{m1}) A_1^\dagger A_1\otimes B_1^\dagger B_1 = I\otimes I.
\end{equation}
This implies that both $A_1^\dagger A_1^{}$ and $B_1^\dagger B_1^{}$ are
proportional to the identity, so $A_1$ and $B_1$ are proportional to unitary
operators, and of course the same argument works for every $m$.  Since local
unitaries cannot change the Schmidt coefficients, it is obvious that
$\ket{\psi}$ and $\ket{\phi}$ must share the same set of Schmidt coefficients,
that is $\lambda_j=\mu_j$, for every $j$, and this proves (iii).

To prove (iv), note that if there is a separable operation carrying
$\ket{\psi}$ to $\ket{\phi}$ and another carrying $\ket{\phi}$ to
$\ket{\psi}$, the Schmidt ranks of $\ket{\psi}$ and $\ket{\phi}$ must be equal
by (i), and \eqref{eqn5} is an equality, so (iii) implies equal Schmidt
coefficients.

Next let us consider the modifications needed when the Schmidt ranks of
$\ket{\psi}$ and $\ket{\phi}$ might be unequal, and are possibly less than the
dimensions of $\HC_A$ or $\HC_B$, which need not be the same.  As noted
previously, \eqref{eqn9} shows that the Schmidt rank of $\ket{\phi}$ cannot be
greater than that of $\ket{\psi}$.  If it is less, then the right side of
\eqref{eqn5} is zero, because at least one of the $\mu_j$ in the product will
be zero, so part (ii) of the theorem is automatically satisfied, part (iii)
does not apply, and (iv) is trivial. Thus we only need to discuss the case in
which the Schmidt ranks of $\ket{\psi}$ and $\ket{\phi}$ have the same value
$r$.  Let $P_A$ and $P_B$ be the projectors on the $\HC_A$ and $\HC_B$
supports $\SC_A$ and $\SC_B$ of $\ket{\psi}$ (as defined at the beginning of
this section), and let $\TC_A$ and $\TC_B$ be the corresponding supports of
$\ket{\phi}$. Note that each of these subspaces is of dimension $r$.  Since
$(P_A\otimes P_B)\ket{\psi}=\ket{\psi}$, \eqref{eqn8} can be rewritten as
\begin{equation}
\label{eqn16}
\bigl(A'_m\otimes B'_m\bigr)\ket{\psi}=\sqrt{p_m}\ket{\phi},
\end{equation}
where
\begin{equation}
\label{eqn17}
 A'_m =  A_m P_A,\quad B'_m =  B_m P_B
\end{equation}
are the operators $A_m$ and $B_m$ restricted to the supports of $\ket{\psi}$.
In fact, $A'_m$ maps $\SC_A$ onto $\TC_A$, and $B'_m$ maps $\SC_B$ onto
$\TC_B$, as this is the only way in which \eqref{eqn16} can be satisfied
when $\ket{\phi}$ and $\ket{\psi}$ have the same Schmidt rank. Finally,
by multiplying \eqref{eqn2} by $P_A\otimes P_B$ on both left and right
one arrives at the closure condition
\begin{equation}
\label{eqn18}
\sum_m {A'_m}^\dagger {A'_m}^{}\otimes {B'_m}^\dagger {B'_m}^{}=P_A\otimes P_B.
\end{equation}
Thus if we use the restricted operators $A'_m$ and $B'_m$ we are back to the
situation considered previously, with $\SC_A$ and $\TC_A$ (which are
isomorphic) playing the role of $\HC_A$, and $\SC_B$ and $\TC_B$ the role of
$\HC_B$, and hence the previous proof applies.
\end{proof}

Some connections between LOCC and the more general category of separable
operations are indicated in the following corollaries:

\begin{corollary}
\label{crl1}
When $\ket{\psi}$ is majorized by $\ket{\phi}$, so there is a deterministic
LOCC mapping $\ket{\psi}$ to $\ket{\phi}$, there does not exist a separable
operation that deterministically maps $\ket{\phi}$ to $\ket{\psi}$, unless
these have equal Schmidt coefficients (are equivalent under local unitaries).
\end{corollary}

This is nothing but (iv) of Theorem 1 applied when the $\ket{\psi}$ to
$\ket{\phi}$ map is LOCC, and thus separable.  It is nonetheless worth
pointing out because majorization provides a very precise characterization of
what deterministic LOCC operations can accomplish, and the corollary
provides a connection with more general separable operations.

\begin{corollary}
\label{crl2}
If either $\HC_A$ or $\HC_B$ is 2-dimensional, then $\ket{\psi}$ can be
deterministically transformed to $\ket{\phi}$ if and only if this is possible
using LOCC, i.e., $\ket{\psi}$ is majorized by $\ket{\phi}$.
\end{corollary}

The proof comes from noting that when there are only two nonzero Schmidt
coefficients, the majorization condition is $\mu_0\geq \lambda_0$, and this is
equivalent to \eqref{eqn5}.

\section{Separable random unitary channel}
\label{sct3}

\subsection{Condition for deterministic mapping}
\label{sct3a}

Any quantum operation (trace-preserving completely positive map) can be
thought of as a quantum channel, and if the Kraus operators are proportional
to unitaries, the channel is bistochastic (maps $I$ to $I$) and is called a
random unitary channel or a random external field in Sec.~10.6
of \cite{GeometryQuantumStates}.  Thus a separable operation in which the $A_m$
and $B_m$ are proportional to unitaries $U_m$ and $V_m$, so \eqref{eqn1} takes
the form
\begin{equation}
\label{eqn19}
\rho' = \Lambda(\rho)=\sum_m p_m\big(U_m\otimes
V_m\big)\rho\big(U_m\otimes V_m\big)^\dagger,
\end{equation}
with the $p_m>0$ summing to 1, can be called a separable random unitary
channel.  We shall be interested in the case in which $\HC_A$ and $\HC_B$ have
the same dimension $d$, and in which the separable unitary channel
deterministically maps not just one but a collection $\{\ket{\psi_j}\}$,
$1\leq j\leq N$ of pure states of full Schmidt rank $d$ to pure states.  This
means that \eqref{eqn8} written in the form
\begin{equation}
\label{eqn20}
\bigl(U_m \otimes V_m\bigr)
\ket{\psi_j} \doteq \ket{\phi_j},
\end{equation}
must hold for all $j$ as well as for all $m$.  The dot equality $\doteq$ means
the two sides can differ by at most a complex phase.  Here such phases cannot
simply be incorporated in $U_m$ or $V_m$, because \eqref{eqn20} must hold for
all values of $j$, even though they are not relevant for the map carrying
$\dyad{\psi_j}{\psi_j}$ to $\dyad{\phi_j}{\phi_j}$.

\begin{theorem}
\label{thm2}
Let $\{\ket{\psi_j}\}$, $1\leq j\leq N$ be a collection of states of full
Schmidt rank on a tensor product $\HC_A\otimes\HC_B$ of two spaces of equal
dimension, and let $\Lambda$ be the separable random unitary channel defined by
\eqref{eqn19}. Let $\psi_j$ and $\phi_j$  be the operators dual to
$\ket{\psi_j}$ and $\ket{\phi_j}$---see \eqref{eqn6} and \eqref{eqn7}.

i) If every $\ket{\psi_j}$ from the collection is deterministically mapped to
a pure state, then
\begin{equation}
\label{eqn21}
 U_m^\dagger U_n^{} \psi_j^{}\psi_k^\dagger \doteq
 \psi_j^{}\psi_k^\dagger U_m^\dagger U_n^{}
\end{equation}
for every $m, n, j,$ and $k$.

ii) If \eqref{eqn21} holds for a \emph{fixed} $m$ and every $n, j,$ and $k$,
it holds for every $m, n, j,$ and $k$.  If in addition \emph{at least one} of
the states from the collection $\{\ket{\psi_j}\}$ is deterministically mapped
to a pure state by $\Lambda$, then every state in the collection is mapped to
a pure state.

iii) Statements (i) and (ii) also hold when \eqref{eqn21} is replaced
with
\begin{equation}
\label{eqn22}
 V_m^\dagger V_n^{} \psi_j^\dagger\psi_k^{} \doteq
 \psi_j^\dagger\psi_k^{} V_m^\dagger V_n^{}.
\end{equation}

\end{theorem}

\begin{proof}

Part (i).
By map-state duality \eqref{eqn20} can be rewritten as
\begin{equation}
\label{eqn23}
 U_m\psi_j\bar V_m \doteq \phi_j,
\end{equation}
where $\bar V_m$ is the transpose of $V_m$---see the remarks following
\eqref{eqn9}. By combining \eqref{eqn23} with its adjoint with $j$ replaced by
$k$, and using the fact that $\bar V_m$ is unitary, we arrive at
\begin{equation}
\label{eqn24}
 U_m^{}\psi_j^{}\psi_k^\dagger U_m^\dagger \doteq \phi_j^{}\phi_k^\dagger.
\end{equation}
Since the right side is independent of $m$, so is the left, which means that
\begin{equation}
\label{eqn25}
 U_n^{}\psi_j^{}\psi_k^\dagger U_n^\dagger  \doteq
 U_m^{}\psi_j^{}\psi_k^\dagger U_m^\dagger.
\end{equation}
Multiply on the left by $U_m^\dagger$ and on the right by $U_n$ to obtain
\eqref{eqn21}.

Part (ii).
If \eqref{eqn25}, which is equivalent to \eqref{eqn21}, holds for $m=1$ it
obviously holds for all values of $m$. Now assume that $\ket{\psi_1}$ is
mapped by $\Lambda$ to a pure state $\ket{\phi_1}$, so \eqref{eqn23} holds
for all $m$ when $j=1$. Take the adjoint of this equation and multiply by
$\bar V_m$ to obtain
\begin{equation}
\label{eqn26}
 \psi_1^\dagger U_m^\dagger \doteq \bar V_m^{}\phi_1^\dagger.
\end{equation}
Set $k=1$ in \eqref{eqn25}, and use  \eqref{eqn26} to rewrite it as
\begin{equation}
\label{eqn27}
 U_n^{}\psi_j^{}\bar V_n^{}\phi_1^\dagger \doteq
 U_m^{}\psi_j^{}\bar V_m^{}\phi_1^\dagger.
\end{equation}
Since by hypothesis $\ket{\psi_1}$ has Schmidt rank $d$, the same is true of
$\psi_1$, and since $U_m$ and $\bar V_m$ in \eqref{eqn23} are unitaries,
$\phi_1$ and thus also $\phi_1^\dagger$ has rank $d$ and is invertible.
Consequently, \eqref{eqn27} implies that
\begin{equation}
\label{eqn28}
 U_n^{}\psi_j^{}\bar V_n^{} \doteq U_m^{}\psi_j^{}\bar V_m^{},
\end{equation}
and we can define $\phi_j$ to be one of these common values, for example
$U_1\psi_j\bar V_1$.  Map-state duality transforms this $\phi_j$
into $\ket{\phi_j}$ which, because of \eqref{eqn28}, satisfies \eqref{eqn20}.

Part (iii). The roles of $U_m$ and $V_m$ are obviously symmetrical, but our
convention for map-state duality makes $\psi_j$ a map from $\HC_B$ to $\HC_A$,
which is the reason why its adjoint appears in \eqref{eqn22}.
\end{proof}

\subsection{Example}
\label{sct3b}

Let us apply Theorem \ref{thm2} to see what pure states of full Schmidt rank
are deterministically mapped onto pure states by the following separable
random unitary channel on two qubits:
\begin{equation}
\label{eqn29}
\Lambda(\rho)=p\rho+(1-p)(X\otimes Z)\rho (X\otimes Z).
\end{equation}
The Kraus operators are $I\otimes I$ and $X\otimes Z$, so $U_1=I$ and $U_2=X$.
Thus the condition \eqref{eqn21} for a collection of states $\{\ket{\psi_j}\}$
to be deterministically mapped to pure states is
\begin{equation}
\label{eqn30}
X\psi_j\psi_k^\dagger\doteq \psi_j\psi_k^\dagger X.
\end{equation}
It is easily checked that
\begin{equation}
\label{eqn31}
 \ket{\psi_1}= (\ket{+}\ket{0}+\ket{-}\ket{1})/\sqrt{2}
\end{equation}
 is mapped to itself by \eqref{eqn29}.  If the corresponding
\begin{equation}
\label{eqn32} \psi_1= \frac{1}{2}\left(
\begin{array}{cc}
1 & 1\\
1 & -1\\
\end{array}
\right)
\end{equation}
is inserted in \eqref{eqn30} with $k=1$, one can show that \eqref{eqn30} is
satisfied for any $2\times2$ matrix
\begin{equation}
\label{eqn33}
\psi_j=\left(
\begin{array}{cc}
a_j & b_j\\
c_j & d_j\\
\end{array}
\right)
\end{equation}
having $c_j=\pm a_j$ and $d_j=\mp b_j$, and that in turn these
satisfy \eqref{eqn30} for every $j$ and $k$. Thus all states of the form
\begin{equation}
\label{eqn34}
\ket{\psi_\pm}=a\ket{00}+b\ket{01}\pm a\ket{10}\mp b\ket{11}
\end{equation}
with $a$ and $b$ complex numbers, are mapped by this channel into pure states.

\section{Conclusions}
\label{sct4}

Our main results are in Theorem~\ref{thm1}: if a pure state on a bipartite
system $\HC_A\otimes\HC_B$ is deterministically mapped to a pure state by a
separable operation $\{A_m\otimes B_m\}$, then the product of the Schmidt
coefficients can only decrease, and if it remains the same, the two sets of
Schmidt coefficients are identical to each other, and the $A_m$ and $B_m$
operators are proportional to unitaries.  (See the detailed statement of the
theorem for situations in which some of the Schmidt coefficients vanish.)
This \emph{product condition} is necessary but not sufficient: i.e., even if
it is satisfied there is no guarantee that a separable operation exists which
can carry out the specified map.  Indeed, we think it is likely that when both
$\HC_A$ and $\HC_B$ have dimension 3 or more there are situations in which the
product condition is satisfied but a deterministic map is not possible.  The
reason is that \eqref{eqn5} is consistent with $\ket{\phi}$ having a larger
entanglement than $\ket{\psi}$, and we doubt whether a separable operation can
increase entanglement.  While it is known that LOCC cannot increase the
average entanglement [\cite{HorodeckiQuantumEntanglement}, Sec.~XV D], there
seems to be no similar result for general separable operations. This is an
important open question.

It is helpful to compare the product condition \eqref{eqn5} with Nielsen's
majorization condition, which says that a deterministic separable operation of
the LOCC type can map $\ket{\psi}$ to $\ket{\phi}$ if and only if $\ket{\phi}$
majorizes $\ket{\psi}$ \cite{ntk1}. Corollary~\ref{crl2} of Theorem~\ref{thm1}
shows that the two are identical if system $A$ or system $B$ is 2-dimensional.
Under this condition a general separable operation can deterministically map
$\ket{\psi}$ to $\ket{\phi}$ only if it is possible with LOCC.  This
observation gives rise to the conjecture that when either $A$ or $B$ is
2-dimensional \emph{any} separable operation is actually of the LOCC form.
This conjecture is consistent with the fact that the well-known example
\cite{PhysRevA.59.1070} of a separable operation that is \emph{not} LOCC uses
the tensor product of two 3-dimensional spaces.  But whether separable and LOCC
coincide even in the simple case of a $2\times 2$ system is at present an
open question (see note added in proof).

When the dimensions of $A$ and $B$ are both 3 or more the product condition of
Theorem~\ref{thm1} is weaker than the majorization condition: if $\ket{\phi}$
majorizes $\ket{\psi}$ then \eqref{eqn5} will hold \cite{ntk2}, but the
converse is in general not true. Thus there might be situations in which a
separable operation deterministically maps $\ket{\psi}$ to $\ket{\phi}$ even
though $\ket{\phi}$ does not majorize $\ket{\psi}$.  If such cases exist,
Corollary~\ref{crl1} of Theorem~\ref{thm1} tells us that $\ket{\psi}$ and
$\ket{\phi}$ must be incomparable under majorization: neither one majorizes the
other.  Finding an instance, or demonstrating its impossibility, would help
clarify how general separable operations differ from the LOCC subclass.

When a separable operation deterministically maps $\ket{\psi}$ to $\ket{\phi}$
and the product of the two sets of Schmidt coefficients are the same, part
(iii) of Theorem~\ref{thm1} tells us that the collections of Schmidt
coefficients are in fact identical, and that the $A_m$ and $B_m$ operators
(restricted if necessary to the supports of $\ket{\psi}$) are proportional to
unitaries.  Given this proportionality (and that the map is deterministic),
the identity of the collection of Schmidt coefficients is immediately evident,
but the converse is not at all obvious.  The result just mentioned can be used
to simplify part of the proof in some interesting work on local copying,
specifically the unitarity of local Kraus operators in [\cite{AnselmiChefles:LocalCopying}, Sec.~3.1]. It might have applications in other cases
where one is interested in deterministic nonlocal operations.

Finally, Theorem~\ref{thm2} gives conditions under which a separable random
unitary operation can deterministically map a whole collection of pure states
to pure states.  These conditions [see \eqref{eqn21} or \eqref{eqn22}] involve
both the unitary operators and the states themselves, expressed as operators
using map-state duality, in an interesting combination.  While these results
apply only to a very special category, they raise the question whether
simultaneous deterministic maps of several pure states might be of interest
for more general separable operations. The nonlocal copying problem, as
discussed in
\cite{AnselmiChefles:LocalCopying,KayEricsson:LocalCloning,GhoshKar:LocalCloning,OwariHayashi:LocalCopying},
is one situation where results of this type are relevant, and there may be
others.

\textit{Note added in proof.} Our conjecture on the equivalence of separable operations and LOCC for low dimensions has been shown to be false \cite{quantph.0705.0795}.

\begin{acknowledgments}
  We thank Shengjun Wu for useful conversations. The research described here
  received support from the National Science Foundation through Grant No.
  PHY-0456951.
\end{acknowledgments}


\begin{thebibliography}{17}
\bibitem[{ntk({\natexlab{a}})}]{ntk1}
\bibinfo{note}{By ``$\ket{\phi}$ majorizes $\ket{\psi}$'' we mean that the
  vector of eigenvalues of the reduced density operator $\rho(\phi)$ of
  $\ket{\phi}$ on $\HC_A$ majorizes that of the reduced density operator
  $\rho(\psi)$ of $\ket{\psi}$ in the sense discussed in
  \cite{Nielsen:Majorization}, or in [\cite{NielsenChuang:QuantumComputation}, Sec.~12.5.1]: the sum of the $k$ largest
  eigenvalues of $\rho(\phi)$ is never smaller than the corresponding sum for
  $\rho(\psi)$. A helpful discussion of majorization is also found in
  \cite{HornJohnson:MatrixAnalysis} (see the index), with, however, the
  opposite convention from Nielsen for ``$A$ majorizes $B$''}.
\bibitem[{ntk({\natexlab{b}})}]{ntk2}
\bibinfo{note}{The general argument that \eqref{eqn5} is implied by (though it
  does not imply) majorization will be found in [\cite{MajorizationNotes}, Sec.~4], or as an exercise on [\cite{HornJohnson:MatrixAnalysis}, p.~199]}.
\end{thebibliography}

\end{document}